\pdfoutput=1
\newif\ifFull
\Fulltrue
\ifFull
\documentclass[11pt]{llncs}
\else
\documentclass{llncs}
\fi

\usepackage{algorithmicx}
\usepackage{algcompatible}
\usepackage[boxed]{algorithm}
\usepackage{tikz}
\usetikzlibrary{%
  arrows,
  calc,
  chains,
  intersections,
  scopes,
  snakes,
  decorations.markings
}

\renewcommand{\algorithmicensure}{\textbf{Invariant:}}
\renewcommand{\algorithmicrequire}{\textbf{Input:}}

\tikzset{
  on each segment/.style={
    decorate,
    decoration={
      show path construction,
      moveto code={},
      lineto code={
        \path [#1]
        (\tikzinputsegmentfirst) -- (\tikzinputsegmentlast);
      },
      curveto code={
        \path [#1] (\tikzinputsegmentfirst)
        .. controls
        (\tikzinputsegmentsupporta) and (\tikzinputsegmentsupportb)
        ..
        (\tikzinputsegmentlast);
      },
      closepath code={
        \path [#1]
        (\tikzinputsegmentfirst) -- (\tikzinputsegmentlast);
      },
    },
  },
  mid arrow/.style={postaction={decorate,decoration={
        markings,
        mark=at position .5 with {\arrow[#1]{angle 60}}
      }}},
}

\begin{document}

\title{Streamed Graph Drawing and\\ the File Maintenance Problem}

\author{Michael T. Goodrich
\and Pawe\l{} Pszona}

\institute{
Dept.~of Computer Science, University of California, Irvine
}
\date{}

\maketitle
\ifFull\else
\pagestyle{plain}
\fi

\begin{abstract}
In \emph{streamed graph drawing}, a planar graph, $G$, is given incrementally
as a data stream 
and a straight-line drawing of $G$ must be updated after each new edge is released.
To preserve the mental map, changes to the drawing
should be minimized after each update, and
Binucci~{\it et al.}~show that exponential area is necessary and
sufficient for a number of streamed graph drawings
for trees if edges are not allowed to move at all.
We show that a number of streamed graph drawings can, in fact,
be done with polynomial area,
including planar streamed graph drawings
of trees, tree-maps, and outerplanar graphs,
if we allow for a small
number of coordinate movements after each update.
Our algorithms involve an interesting connection to a classic algorithmic
problem---the \emph{file maintenance problem}---and we also give new algorithms
for this problem in a framework where bulk memory moves are allowed.
\end{abstract}

\section{Introduction}

In the \emph{streamed graph drawing} framework, which was introduced by
Binucci~{\it et al.}~\cite{DBLP:conf/gd/BinucciBBDGPPSZ09,Binucci2012418},
a graph, $G$, is incrementally
presented as a data stream of its vertices and edges, 
and a drawing of $G$ must be updated after each new 
edge is released.
So as to preserve 
the \emph{mental map}~\cite{Eades1991,Misue1995183} of the drawing,
this framework also requires that changes to the drawing of $G$ 
should be minimized after each update.
Indeed, to achieve this goal,
Binucci~{\it et al.} took the extreme position of requiring that
once an edge is drawn no changes can be made to that edge.
They showed that, under this restriction, exponential area is necessary 
and sufficient for planar drawings of trees under various orderings for how
the vertices and edges of the trees are presented.

In light of recent results regarding the mental map~\cite{MentalMap},
however, we now know that moving a small number of vertices
or edges in a drawing of a graph 
does not significantly affect readability in a negative way. Therefore,
in this paper,
we choose to relax the requirement that there are no changes to the drawing of
the graph after an update 
and instead allow a small number of coordinate movements
after each such update.
In this paper, we study planar streamed graph drawing schemes for trees,
tree-maps, and outerplanar graphs, showing that polynomial area is achievable
for such streamed graph drawings if small changes to the drawings are allowed
after each update.
Our results are based primarily on an
interesting connection between streamed graph drawing and a classic algorithmic
problem, the \emph{file maintenance problem}.

In the \emph{file maintenance problem}~\cite{DBLP:conf/sigmod/Willard86},
we wish to maintain an ordered
set, $S$, of $n$ elements, such that each element,
$x$ in $S$, is assigned a unique integer label, $L(x)$, in the range $[0,N]$,
where $x$ comes before $y$ if and only if $L(x)<L(y)$.
In the classic version of this problem, $N$ is restricted to be $O(n)$, with the
motivation that the integer labels are addresses or pseudo-addresses 
for memory locations where the elements of $S$ are stored\footnote{For instance,
	in the EDT text editor developed for the DEC PDP-11 series 
        of minicomputers,
	each line was assigned a pseudo line number, 1.0000, 2.0000, 
        and so on, and
	if a new line was to be introduced between two existing lines, $x$ and $y$,
	it was given as a default label
        the average of the labels of $x$ and $y$ as its label.}.
If $N$ is only restricted to be polynomial
in $n$, then this is known as the 
\emph{online list labeling 
problem}~\cite{DBLP:conf/esa/BenderCDFZ02,Dietz:1987:TAM:28395.28434}.
In either case, the set, $S$, can be updated by issuing a command,
insertAfter$(x,y)$, where $y$ is to be inserted to be immediately after $x\in S$ in
the ordering,
or
insertBefore$(x,y)$, where $y$ is to be inserted to be immediately 
before $x\in S$ in the ordering.
The goal is to minimize the number of elements in $S$ needing to be relabeled as a
result of such an update.

\paragraph{\textbf{Previous Related Results}.}
For the file maintenance problem, 
Willard~\cite{DBLP:conf/sigmod/Willard86}
gave a rather complicated
 solution that achieves $O(\log^2n)$ relabelings in the worst case 
after each insertion, and this result was later
simplified by Bender~{\it et al.}~\cite{DBLP:conf/esa/BenderCDFZ02}.
For the online list labeling problem,
Dietz and Sleator~\cite{Dietz:1987:TAM:28395.28434} give
an algorithm that achieves $O(\log n)$ amortized relabelings per insertion,
and $O(\log^2 n)$ in the worst-case,
using an $O(n^2)$ tag range. Their solution was simplified by
Bender~{\it et al.}~\cite{DBLP:conf/esa/BenderCDFZ02} with the same bounds.
Recently,
Kopelowitz~\cite{DBLP:conf/focs/Kopelowitz12} has given an 
algorithm that achieves $O(\log n)$ worst-case
relabelings after each insertion, using a polynomial bound for $N$.

For streamed tree drawings, as we mentioned above,
Binucci~{\it et al.}~\cite{DBLP:conf/gd/BinucciBBDGPPSZ09,Binucci2012418}
show that exponential area is required for planar drawings of trees, depending
on the order in which vertices and edges are introduced (e.g, BFS, DFS, etc.).

\paragraph{\textbf{Our Results}.}
For the context of this paper, we focus on planar drawings of graphs, 
so we consider a drawing to consist essentially of a set of non-crossing line
segments.
For traditional drawings of trees and outerplanar graphs, the endpoints of the
segments correspond to vertices and the segments represent edges. In tree-map
drawings, each vertex $v$ of a tree $T$ is represented by a rectangle,
$R$, such that the children of $v$ are represented
by rectangles inside $R$ that share portions of at least two sides of $R$.
Thus, in a tree-map drawing, the line segments correspond to the sides of
rectangles.

We present new streamed graph drawing 
algorithms for general trees, tree-maps, and outerplanar graphs
that keep the area of the drawing to be of polynomial size and allow new edges
to arrive in any order, provided the graph is connected at all times.
After each update to a graph is given, we allow a small number of,
say, a polylogarithmic number of the endpoints of the segments in 
the drawing to move to accommodate the representation of the new edge.
We alternatively
consider these to be movements of either individual endpoints or sets of 
at most $B$ endpoints, for a parameter $B$,
provided that each set of such endpoints is contained in a given convex region,
$R$, and all the endpoints in this region are translated by the same vector.
We call such operations the \emph{bulk moves}.

All of our methods are based on our
showing interesting connections between the streamed graph drawing
problems we study and the file maintenance problem.
In addition to utilizing existing algorithms
for the file maintenance problem in our graph drawing schemes,
we also give a new algorithm
for this classic problem in a framework where bulk memory moves are allowed,
and we show how this solution can also be applied to streamed graph drawing.

\section{Building Blocks}

\paragraph{\textbf{The ordered streaming model}.} 
We start with the description of the model under which we operate.
At each time $t\geq1$, a new edge, $e$,
of a graph, $G$,
arrives and has to be incorporated immediately into a drawing of $G$, 
using line segments whose endpoints are placed at grid points with 
integer coordinates.
Since we are focused on planar drawings in this paper,
together with the edge, $e$,
we also get the information of its relative position among the neighbors
of $e$'s endpoints (i.e., for every vertex we know the clockwise order
of its neighbors and $e$'s placement in this order). 
At all times, the current graph, $G$, is connected, and
the edges never disappear (infinite persistence).

Incidentally, the streaming
model of Binucci~{\it et al.}~\cite{DBLP:conf/gd/BinucciBBDGPPSZ09}
is slightly different, 
in that edges arrive without the order information in their model.
Under that model, they have given an $\Omega(2^\frac{n}{8})$ area lower bound
for binary tree drawing and an $\Omega(n(d-1)^n)$ lower bound for drawing
trees with maximum degree $d>2$. These bounds stand when the algorithm is not
allowed to move any vertex. However, they only apply to a very restricted
class of algorithms, namely \emph{predefined-location algorithms}, which are
non-adaptive algorithms whose behavior does not depend on the order
in which the edges arrive or the previously drawn edges.
Also, as noted above, Binucci~{\it et al.}~do not allow for 
vertices to move once they are placed.
As we show in the following theorem, even with the added information regarding
the relative placement of an edge among its neighbors incident
on the same  vertex, if we don't allow for vertex moves, we must allow for 
exponential area.

\begin{theorem}
\label{thm-lower}
Under the ordered streaming model without vertex moves, any tree drawing
algorithm requires $\Omega(2^{n/2})$ area in the worst case.
\label{area_lower_bound}
\end{theorem}

\ifFull
\begin{proof}
We start with a single node $r$ (\emph{root}) placed in an empty grid
and imagine it surrounded by an $2\times 2$ \emph{bounding square}
(see Fig.~\ref{wedges}).
After adding at most 5 edges from the root (ordering doesn't matter),
there has to be one side $m$ of the \emph{bounding square} that is pierced
by two edges, $e$ and $f$. Assume w.l.o.g. that $m$ is the upper side
of the \emph{bounding square}.
$e$ and $f$ form an (infinite) wedge $W$, and their points of intersection
with $m$ specify $s$, a sub-interval of $m$. Abusing the notation so
that $s$ denotes the length of $s$, we have $s \leq 2$.

We keep adding new edges from the root \emph{inside} (shrinking) $W$.
Adding new edge inside $W$ (enforced by specifying edge ordering)
divides $W$ into two wedges, $W_1$ and $W_2$.
For the next iteration we select the one with the smaller interval $s$.
Thus, after adding $n+1$ edges from the root, $s \leq 1/2^n$
in the current wedge.

Now let us limit $W$ with a lowest possible horizontal line $L$
such that there are two
grid points inside $W$ that lie on $L$.
There are two triangles cut out from $W$ -- one (red) by the
\emph{bounding square},
the other (blue) by $L$. They are similar, so $A/S = a/s$. Then

\begin{displaymath}
A = \frac{aS}{s} \geq 2^naS \geq 2^n,
\end{displaymath}
as $a \geq 1$ and $S \geq 1$.

Now look at the large (blue) triangle. By definition, it can
contain at most $A$ grid points (at most one for each $y$-coordinate;
distance of $r$ from $L$ is $\leq A$).
Placing of $\log A \approx n$ additional edges allows us to obtain
a sub-wedge with no grid points inside the blue triangle (by always
picking the sub-wedge that contains less grid points).
Now every new edge has to be placed at distance
at least $A$ from the root as there are no more available grid points
left inside the triangle. Clearly, the area needed for the drawing is now also
at least $A \geq 2^n$. Since our tree has $2n$ edges, we get an
$\Omega(2^{n/2})$ area lower bound for drawing a tree of $n$ edges.
\qed
\end{proof}

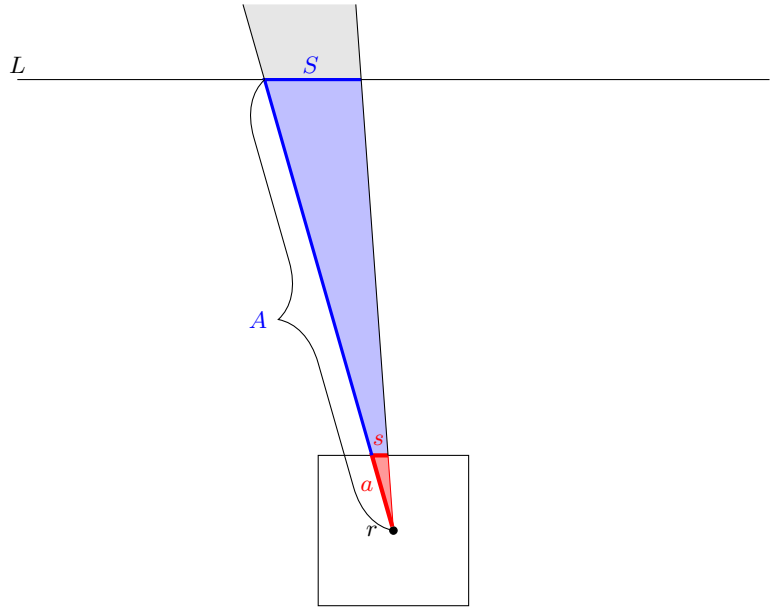
\begin{figure}[ht]
  \begin{center}
    \begin{tikzpicture}
      [
        n/.style={circle,fill=black,draw,inner sep=1pt},
        node distance=1pt
      ]

      \node[n] (r) at (0,0) {};
      \node[left=of r] {$r$};
      \coordinate (e3) at (-2,7);
      \coordinate (f3) at (-0.5,7);
      \fill[black!10] (e3) -- (r) -- (f3) -- (e3);
      \draw (-1,1) -- (1,1) -- (1,-1) -- (-1,-1) -- cycle;
      \draw (-5,6) -- (5,6);
      \path[name path=e] (r) -- (e3);
      \path[name path=f] (r) -- (f3);
      \path[name path=L] (-5,6) -- (5,6);
      \path[name path=side] (-5,1) -- (5,1);
      \path[name intersections={of=e and L, by=e2}];
      \path[name intersections={of=f and L, by=f2}];
      \path[name intersections={of=e and side, by=e1}];
      \path[name intersections={of=f and side, by=f1}];
      \fill[blue!25] (e2) -- (r) -- (f2) -- (e2);
      \fill[red!40] (e1) -- (r) -- (f1) -- (e1);
      \draw[ultra thick,red] (r) -- (e1) -- (f1);
      \draw[very thick,blue] (e1) -- (e2) -- (f2);
      \draw[decorate,decoration={brace,amplitude=20pt}] (r.center) -- (e2);
      \draw (e2) -- (e3);
      \draw[red] (r) -- (e1);
      \draw[red] (r) -- (f1);
      \draw (f1) -- (f3);
      \node[red] at (-0.35,0.6) {$a$};
      \node[red] at (-0.2,1.2) {$s$};
      \node[blue] at (-1.8, 2.8) {$A$};
      \node[blue] at (-1.1, 6.2) {$S$};
      \node at (-5, 6.2) {$L$};
    \end{tikzpicture}
  \end{center}
  \caption{Illustration of the proof of Theorem~\ref{area_lower_bound}}
  \label{wedges}
\end{figure}

\else
\begin{proof}
See Appendix.
\qed
\end{proof}
\fi

\paragraph{\textbf{File maintenance with bulk moves}.}
Here we consider the file maintenance problem and the online list labeling problem
variants where we allow for bulk moves\footnote{Note that bulk moves are also motivated
  for the original file maintenance problem if we define the complexity
  of a solution in terms of the number of commands that are 
  sent to a DMA controller for bulk memory-to-memory moves.}
of an interval of $B$ labels, for some parameter $B$.
We have already mentioned the known results for the file maintenance problem, 
where the only type of relabelings that are allowed are for individual elements,
in which $O(\log^2 n)$ worst-case 
relabelings suffice for each update when $N$ is 
$O(n)$~\cite{DBLP:conf/esa/BenderCDFZ02,DBLP:conf/sigmod/Willard86} and 
$O(\log n)$ suffice in the worst-case when $N$ is a polynomial 
in $n$~\cite{DBLP:conf/focs/Kopelowitz12}.

Bulk moves allow us to improve on these bounds.
We have achieved several tradeoffs between the operation count and the size of $B$.
Of course, if $B$ is $n$, then achieving constant number of operations is easy,
since we can maintain the $n$ elements to have the indices $1$ to $n$, and with each
insertion, at some rank $i$, we simply move the elements currently from $i$ to
$n$ up by one, as a single bulk move. Theorem~\ref{tradeoff_thm} summarizes the
rest of our results.

\begin{theorem}
We can achieve the following bounds:
\begin{enumerate}
\item
$O(1)$ worst-case relabeling bulk moves suffice for the file maintenance
problem if $B=n^{1/2}$.
\item
$O(1)$ worst-case relabeling bulk moves suffice for the online list maintenance
problem if $B=\log n$.
\item
$O(\log n)$ worst-case relabeling bulk moves suffice for the file maintenance
problem if $B=\log n$.
\end{enumerate}
\label{tradeoff_thm}
\end{theorem}
\begin{proof}
\leavevmode
\begin{enumerate}
\item
This is accomplished in an amortized way by partitioning
the array into $n^{1/2}$ chunks of size at most $2n^{1/2}$.
Whenever a chunk, $i$, grows to have size $n^{1/2}$, we move all the chunks to
the right of $i$ 
by one chunk (using $O(n^{1/2})$ bulk moves). Then we split the chunk $i$ in
two, keeping half the items in chunk $i$ and moving half to chunk $i+1$.
These moves are charged to the previous $n^{1/2}/2$ insertions in chunk $i$.
Turning this bound into a worst-case bound is then done using standard
de-amortization techniques.
\item
This is accomplished by slightly modifying a two-level structure
of Kopelowitz~\cite{DBLP:conf/focs/Kopelowitz12}.
Kopelowitz used the top level of this structure to maintain order
of sublists of size $O(\log n)$ each. Order within each sublist was maintained
using standard file maintenance problem solutions.
Our modification is that each sublist is now represented as a small subarray
of size $O(\log n)$ and operations on the top level of Kopelowitz's structure
are simulated using bulk moves on a big array containing all concatenated subarrays.
\item
This is accomplished by using the method of 
Bender~{\it et al.}~\cite{DBLP:conf/esa/BenderCDFZ02}
and noting that each insertion in their scheme uses a process where
each substep involves moving a contiguous subarray of size $O(\log n)$ using
$O(\log n)$ individual moves. Each such move can alternatively be done using
$O(1)$ bulk moves of subarrays of size $\log n$.
\end{enumerate}
\end{proof}

\section{Streamed Graph Drawing of Trees}

In this section,
we present several algorithms for upward grid drawings of trees in the
ordered streaming model. The algorithms are designed to accommodate
different types of vertex moves allowed. 
For example, by a \emph{bulk move}, we
mean a move that translates all segment endpoints that belong to a given
convex region, $R$, by the same vector. This corresponds to the
observation~\cite{yantis1992multielement}
that moving a small number of elements in the same direction
is easy to follow and does not interfere with the ability to
understand the structure of the drawing (as long as there are no
intersections).


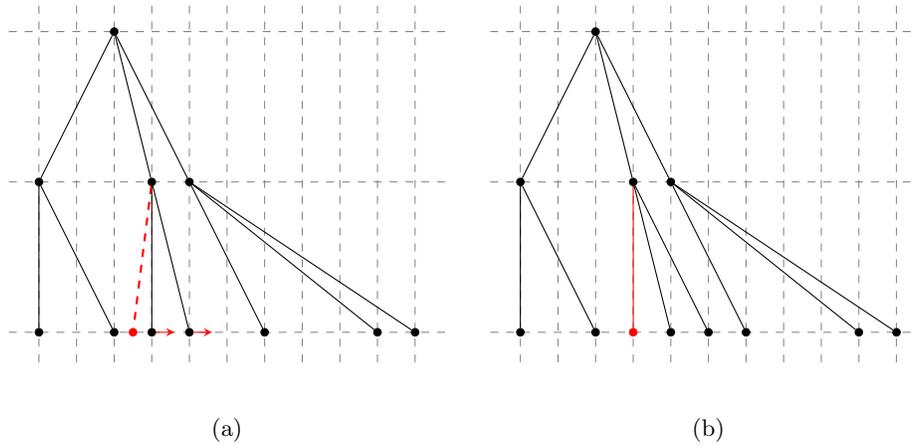
\begin{figure}[htb]
  \begin{center}
    \begin{tikzpicture}
      [
        n/.style={circle,fill=black,draw,inner sep=1pt}
      ]
      \node at (0,0) {};
      \node at (0,5.25) {};

      \draw[xstep=0.5,ystep=18,gray,very thin,dashed] (0,0.6) grid (5,5.4);
      \draw[dashed,gray,very thin] (-0.4,1) -- (5.4,1);
      \draw[dashed,gray,very thin] (-0.4,3) -- (5.4,3);
      \draw[dashed,gray,very thin] (-0.4,5) -- (5.4,5);

      \node[n] (root) at (1,5) {};
      \node[n] (a1) at (0,3) {};
      \node[n] (a2) at (1.5,3) {};
      \node[n] (a3) at (2,3) {};
      \node[n] (b1) at (0,1) {};
      \node[n] (b2) at (1,1) {};
      \node[n] (b3) at (1.5,1) {};
      \node[n] (b4) at (2,1) {};
      \node[n] (b5) at (3,1) {};
      \node[n] (b6) at (4.5,1) {};
      \node[n] (b7) at (5,1) {};

      \draw (root) -- (a1);
      \draw (root) -- (a2);
      \draw (root) -- (a3);
      \draw (a1) -- (b1);
      \draw (a1) -- (b2);
      \draw (a2) -- (b3);
      \draw (a2) -- (b4);
      \draw (a3) -- (b5);
      \draw (a3) -- (b6);
      \draw (a3) -- (b7);

      \node[n,red] (c) at (1.25,1) {};
      \draw[thick,red,dashed] (a2) -- (c);

      \draw[->,red,>=stealth] (b3) -- ++(0.3,0);
      \draw[->,red,>=stealth] (b4) -- ++(0.3,0);

      \node at (2.5,-0.3) {(a)};

      \begin{scope}[xshift=6.4cm]
        \draw[xstep=0.5,ystep=18,gray,very thin,dashed] (0,0.6) grid (5,5.4);
        \draw[dashed,gray,very thin] (-0.4,1) -- (5.4,1);
        \draw[dashed,gray,very thin] (-0.4,3) -- (5.4,3);
        \draw[dashed,gray,very thin] (-0.4,5) -- (5.4,5);

        \node[n] (root) at (1,5) {};
        \node[n] (a1) at (0,3) {};
        \node[n] (a2) at (1.5,3) {};
        \node[n] (a3) at (2,3) {};
        \node[n] (b1) at (0,1) {};
        \node[n] (b2) at (1,1) {};
        \node[n,red] (c) at (1.5,1) {};
        \node[n] (b3) at (2,1) {};
        \node[n] (b4) at (2.5,1) {};
        \node[n] (b5) at (3,1) {};
        \node[n] (b6) at (4.5,1) {};
        \node[n] (b7) at (5,1) {};

        \draw (root) -- (a1);
        \draw (root) -- (a2);
        \draw (root) -- (a3);
        \draw (a1) -- (b1);
        \draw (a1) -- (b2);
        \draw (a2) -- (b3);
        \draw (a2) -- (b4);
        \draw (a3) -- (b5);
        \draw (a3) -- (b6);
        \draw (a3) -- (b7);
        \draw[red] (a2) -- (c);

        \node at (2.5,-0.3) {(b)};
      \end{scope}
    \end{tikzpicture}
  \end{center}
  \caption{Illustrating an insertion for our tree-drawing scheme: (a) determining
    relative position for the new (dashed, red) edge; (b) tree after edge insert and
    related vertex moves.}
  \label{level_draw}
\end{figure}

Let $G$ be a tree that is revealed one edge at a time, keeping the graph connected.
Algorithm~\ref{bulk_alg} selects one endpoint of the first edge, $r$,
puts it at position $(0,0)$, and produces an upward straight-line grid drawing
of $G$,
level-by-level, with each edge from parent to child pointed downwards.
(If a new edge is ever revealed for the current root, we simply recalibrate what 
we are calling position $(0,0)$ without changing the position of 
the vertices already drawn.)
For the $k$th level, $L_k$, with $n_k$ nodes, we place nodes
in positions (0,$-k$) through ($N$,$-k$) in the order of their parents
(to avoid intersections), where $N\ge n_k$. When a new edge is added, we locate
the position (row and position in the row) of the new node and insert the new
node after its predecessor (or before its successor), shifting other nodes
on this level as needed to make room for the new node. 
(See Fig.~\ref{level_draw}.)
The details are as shown in Algorithm~\ref{bulk_alg}.

\begin{algorithm}[hbt]
  \begin{algorithmic}[1]
    \REQUIRE $e = (a,b)$, the edge to be added; $b$ is the new vertex
    \STATE $k \leftarrow b$'s distance from $r$
    \STATE determine $c$, $b$'s predecessor (or successor) in level $L_k$
    \STATE perform $L_k$.insertAfter($c,b$) (or $L_k$.insertBefore($c,b$)), giving $b$ 
           integer label $L(b)$
    \label{move_line}
    \STATE move vertices whose labels have changed in the previous step
    \STATE place $b$ at $(L(b),-k)$ and draw $e$
    \label{place_line}
  \end{algorithmic}
  \caption{Generic insertion algorithm for upward straight-line grid streamed 
           tree drawing.}
  \label{bulk_alg}
\end{algorithm}

Drawing the tree in this fashion ensures there are no intersections
(edges connect only vertices in two neighboring levels), even as the
vertices are shifted (relative order of vertices stays the same).
In addition, there are at most $O(n)$ levels, and at most $O(n)$ nodes per level.

\begin{theorem}
Depending on the implementation for the insertBefore and insertAfter methods,
Algorithm~\ref{bulk_alg} maintains a straight-line upward grid drawing
of a tree in the ordered streaming model to have the following possible performance
bounds:
\begin{enumerate}
\item
$O(n^2)$ area and $O(1)$ vertex moves per insertion 
if bulk moves of size $n^{1/2}$ are allowed.
\item
$O(n^2)$ area and $O(\log n)$ vertex moves per insertion
if bulk moves of size $\log n$ are allowed.
\item
$O(n^2)$ area and $O(\log^2 n)$ vertex moves per insertion 
if bulk moves are not allowed.
\item
polynomial area and $O(1)$ vertex moves per insertion 
if bulk moves of size $\log n$ are allowed.
\item
polynomial area and $O(\log n)$ vertex moves per insertion 
if bulk moves are not allowed.
\end{enumerate}
\end{theorem}

\begin{proof}
The claimed bounds follow 
immediately from Theorem~\ref{tradeoff_thm}.
\qed
\end{proof}

Note that $\Omega(n^2)$ area is necessary in the worst case for an 
upward straight-line 
grid drawing of a tree if siblings are always placed on the same level.

\section{Streamed Graph Drawing of Tree-Maps}

A tree-map, $M$, is a visualization technique introduced by 
Shneiderman~\cite{Shneiderman:1992}, which draws a rooted tree, $T$, 
as a collection of nested rectangles.
The root, $r$, of $T$ is associated with a rectangle, $R$,
and if $r$ has $k$ children, then $R$ is partitioned into $k$ sub-rectangles using 
line segments parallel to one of the coordinate axes (say, the $x$-axis), with
each such rectangle associated with one of the children of $r$.
Then, each child rectangle is recursively partitioned using line 
segments parallel to the other coordinate axis (say, the $y$-axis).
(See Fig.~\ref{tree_map}.)

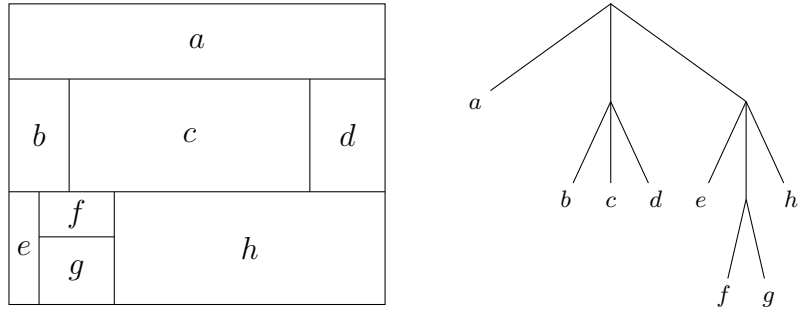
\begin{figure}[hbt!]
  \begin{center}
    \begin{tikzpicture}
      [
        n/.style={inner sep=0pt,outer sep=-0.6},
        level distance=1.3cm,
        level 1/.style={sibling distance=1.8cm,text height=0.2cm},
        level 2/.style={sibling distance=0.6cm},
      ]
      \draw (0,0) -- (0,4) -- (5,4) -- (5,0) -- cycle;

      \draw (0,3) -- (5,3);
      \node at (2.5,3.5) {\large $a$};
      \draw (0,1.5) -- (5,1.5);
      \draw (0.8,1.5) -- (0.8,3);
      \node at (0.4,2.25) {\large $b$};
      \draw (4,1.5) -- (4,3);
      \node at (2.4,2.25) {\large $c$};
      \node at (4.5,2.25) {\large $d$};
      \draw (0.4,0) -- (0.4,1.5);
      \node at (0.2,0.75) {\large $e$};
      \draw (1.4,0) -- (1.4,1.5);
      \draw (0.4,0.9) -- (1.4,0.9);
      \node at (0.9,1.2) {\large $f$};
      \node at (0.9,0.45) {\large $g$};
      \node at (3.2,0.75) {\large $h$};

      \begin{scope}[xshift=8cm]
        \node[n] at (0,4) {}
          child { node {$a$} }
          child {
            child { node {$b$} }
            child { node {$c$} }
            child { node {$d$} }
          }
          child {
            child { node {$e$} }
            child {
              child { node {$f$} }
              child { node {$g$} }
            }
            child { node {$h$} }
          }
        ;
      \end{scope}
    \end{tikzpicture}
  \end{center}
  \caption{A tree-map and its associated tree.}
  \label{tree_map}
\end{figure}

We assume in this case that a tree, $G$, 
is released one edge at a time, as in the previous section.
We assume inductively that we have a tree-map drawn for $G$, with a global set, $X$,
of all $x$-coordinates maintained for the rectangle boundaries and a global set, $Y$,
of all $y$-coordinates maintained for the rectangle boundaries.
When an edge, $e$, of a rectangle has to be moved, the largest
segment containing $e$ is moved accordingly.
Our insertion method is shown in Algorithm~\ref{tree_map_alg}
(for brevity, the case when a vertex has no predecessors among its siblings
is omitted).

\begin{algorithm}[hbt]
  \begin{algorithmic}[1]
    \REQUIRE $e = (a,b)$, the edge to be added; $b$ is a new child vertex
    \STATE Let $R$ be the rectangle for $a$, 
           and let $z$ be the primary axis for $R$ (w.l.o.g., $z=x$)
    \IF {$b$ has no siblings}
       \STATE $R_b \leftarrow R$ (and give it primary axis orthogonal to $z$)
    \ELSIF {}
    \ELSE
    \STATE determine $c$, $b$'s predecessor sibling (w.l.o.g.), and let
	$R_c$ be $c$'s rectangle
    \STATE perform $X$.insertAfter($R_c.x_{\rm max},b$), giving $b$ 
           integer label $L(b)$
    \label{move_line_tm}
    \STATE move segment endpoints whose labels have changed in the previous step
    \STATE $R_b \leftarrow$ the rectangle in $R$ with left 
           boundary $R_c.x_{\rm max}$ and right boundary $L(b)$
    \label{place_line2}
    \ENDIF
  \end{algorithmic}
  \caption{Generic insertion algorithm for streamed 
           tree-map drawing.}
  \label{tree_map_alg}
\end{algorithm}

\begin{theorem}
Depending on the implementation for the insertBefore and insertAfter methods,
Algorithm~\ref{tree_map_alg} maintains a tree-map drawing
of a tree in the ordered streaming model to 
have the following possible performance
bounds:
\begin{enumerate}
\item
$O(n^2)$ area and $O(1)$ $x$- and $y$-coordinate moves per insertion 
if bulk moves of size $n^{1/2}$ are allowed.
\item
$O(n^2)$ area and $O(\log n)$ $x$- and $y$-coordinate moves per insertion
if bulk moves of size $\log n$ are allowed.
\item
$O(n^2)$ area and $O(\log^2 n)$ $x$- and $y$-coordinate  moves per insertion 
if bulk moves are not allowed.
\item
polynomial area and $O(1)$ $x$- and $y$-coordinate moves per insertion 
if bulk moves of size $\log n$ are allowed.
\item
polynomial area and $O(\log n)$ $x$- and $y$-coordinate moves per insertion 
if bulk moves are not allowed.
\end{enumerate}
\end{theorem}

\begin{proof}
The claimed bounds follow immediately Theorem~\ref{tradeoff_thm}.
\qed
\end{proof}

\ifFull
We leave as an exercise how a similar approach could be used for 
streamed grid straight-line drawings 
of binary trees, where the $y$-coordinate of a node depends on its depth and its
$x$-coordinate depends on its inorder rank.
\fi

\section{Streamed Graph Drawing of Outerplanar Graphs}

\ifFull
\begin{figure}[h] 
  \begin{center}
    \begin{tikzpicture}
      [
        n/.style={circle,fill=black,draw,inner sep=1pt},
        node distance=1pt
      ]

      \node[n] (a) at (-5,2) {};
      \node[n] (b) at (-3.5,1.5) {};
      \node[n] (c) at (-5.5,1) {};
      \node[n] (d) at (-3,1) {};
      \node[n] (h) at (-5,0) {};
      \node[n] (g) at (-6,-1) {};
      \node[n] (f) at (-4,-0.75) {};
      \node[n] (e) at (-3,-1.5) {};

      \node[above=of a] {$a$};
      \node[right=of b] {$b$};
      \node[left=of c] {$c$};
      \node[right=of d] {$d$};
      \node[right=of e] {$e$};
      \node[above=of f] {$f$};
      \node[below=of g] {$g$};
      \node[left=of h] {$h$};

      \node at (-4.5, -2.5) {$G$};

      \draw (a) -- (b);
      \draw (a) -- (c);
      \draw (b) -- (c);
      \draw (c) -- (d);
      \draw (c) -- (h);
      \draw (d) -- (h);
      \draw (d) -- (e);
      \draw (e) -- (f);
      \draw (f) -- (g);
      \draw (g) -- (d);

      \path[name path=circ,draw,ultra thin] (1.5,0) circle (2);

      \node[n] (a1) at ($(1.5,0)!1!0:(1.5,2)$) {};
      \node[n] (h1) at ($(1.5,0)!1!45:(1.5,2)$) {};
      \node[n] (g1) at ($(1.5,0)!1!90:(1.5,2)$) {};
      \node[n] (f1) at ($(1.5,0)!1!135:(1.5,2)$) {};
      \node[n] (e1) at ($(1.5,0)!1!180:(1.5,2)$) {};
      \node[n] (d1) at ($(1.5,0)!1!225:(1.5,2)$) {};
      \node[n] (c1) at ($(1.5,0)!1!270:(1.5,2)$) {};
      \node[n] (b1) at ($(1.5,0)!1!315:(1.5,2)$) {};

      \node[above=of a1] {$a$};
      \node[above right=of b1] {$b$};
      \node[right=of c1] {$c$};
      \node[below right=of d1] {$d$};
      \node[below=of e1] {$e$};
      \node[below left=of f1] {$f$};
      \node[left=of g1] {$g$};
      \node[above left=of h1] {$h$};

      \draw (a1) -- (b1);
      \draw (a1) -- (c1);
      \draw (b1) -- (c1);
      \draw (c1) -- (d1);
      \draw (c1) -- (h1);
      \draw (d1) -- (h1);
      \draw (d1) -- (e1);
      \draw (e1) -- (f1);
      \draw (f1) -- (g1);
      \draw (g1) -- (d1);

    \end{tikzpicture}
  \end{center}
  \caption{An outerplanar graph, $G$, and its circular drawing.}
  \label{circ_drawing}
\end{figure}
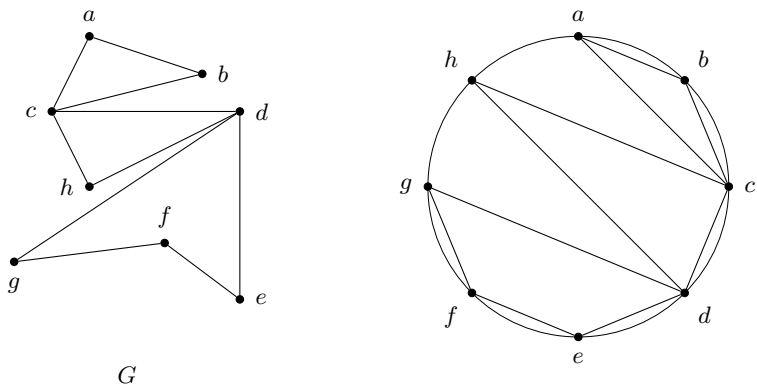 
\fi
Our algorithm for drawing outerplanar graphs in the streaming model is
based on a well-known fact about outerplanar graphs,
namely that any outerplanar graph may be drawn with straight-line edges
and without intersections in such a way that the vertices are placed on
a circle~\cite{Tamassia:2007:HGD:1202383}.
\ifFull
(See Fig.~\ref{circ_drawing}.)
\fi


As previously, we assume that each new edge comes with the information
about its relative placement among its endpoints' incident edges.
In other words, for each vertex, we know the clockwise order
of its incident edges. Fig.~\ref{1to1} shows a situation when
this information alone is not enough, however.

\begin{figure} 
  \vspace*{-12pt}
  \begin{center}
    \begin{tikzpicture}
      [
        scale=0.75,
        n/.style={circle,fill=black,draw,inner sep=1pt},
        node distance=1pt
      ]

      \node[n] (a) at (0, 3) {};
      \node[n] (b) at (0, 1.5) {};
      \node[n] (c) at (0, 0) {};
      \node[n,blue] (d) at (0.5, 1.5) {};

      \node[above=of a] {$a$};
      \node[left=of b] {$b$};
      \node[below=of c] {$c$};
      \node[above=of d] {$d$};

      \draw (a) -- (b) -- (c);

      \draw[green] (a) .. controls (-1,2) and (-1,1) .. (c);
      \draw[dashed,red] (a) .. controls (1,2) and (1,1) .. (c);
      \draw[blue] (b) -- (d);
    \end{tikzpicture}
  \end{center}
  \vspace*{-18pt}
  \caption{Situation where information about order of edges around vertex
  is insufficient. Initially, there are vertices $a$, $b$, $c$ and edges
  $(a,b)$, $(b,c)$. When a new edge $(a,c)$ is added, it can be drawn
  in two ways (solid green or dashed red) -- ordering of edges does not specify which
  one to choose. When edge $(b,d)$ arrives (with edge order $(a,d,c)$
  around $b$), if the dashed red edge location was selected, there is no way to move
  the vertices without intersections to produce an outerplanar drawing.}
  \label{1to1}
\end{figure}
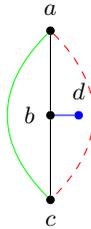 

Nevertheless, this type of problem can only happen when the new edge connects
two vertices of degree 1 as shown below.

\begin{lemma}{If at least one of the newly connected vertices has degree $>1$,
the information about relative order of incident edges suffice.}
\end{lemma}
\begin{proof}
Consider the situation shown in Fig.~\ref{deg2}. $(p,q)$ is
the new edge ($p$ has degree at least 2). The graph is connected,
and the path between $p$ and $q$ is shown. The initial direction
of the edge (bold part) is determined by the ordering of edges around $p$.
Then the edge can either go around $r$ (shown in dashed red) or not (solid green).
Obviously, the dashed red edge location is invalid, as it would surround $r$ with a face,
violating the requirement that each vertex belongs to the outer plane.
Therefore, there is only one possibility for correctly drawing the edge.
\qed
\end{proof}

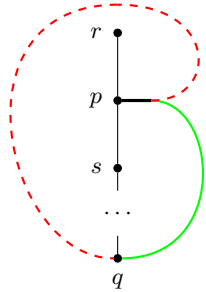
\begin{figure} 
  \vspace*{-12pt}
  \begin{center}
    \begin{tikzpicture}
      [
        scale=0.75,
        n/.style={circle,fill=black,draw,inner sep=1pt},
        node distance=1pt
      ]

      \node[n] (r) at (0,6) {};
      \node[n] (p) at (0,4.8) {};
      \node[n] (p1) at (0,3.6) {};
      \node[n] (q) at (0,2) {};

      \node at (0,2.8) {$\cdots$};

      \node[left=of r] {$r$};
      \node[left=of p] {$p$};
      \node[left=of p1] {$s$};
      \node[below=of q] {$q$};

      \draw (r) -- (p) -- (p1);
      \draw (p1) -- (0,3.2);
      \draw (q) -- (0,2.4);

      \coordinate (start) at (0.6,4.8);

      \draw[very thick](p) -- (start);

      \draw[green,thick] (start) .. controls (1.8, 4.8) and (2.0, 2) .. (q);
      \draw[dashed,red,thick] (start) .. controls (1.8, 4.8) and (1.8, 6.5) .. (0, 6.5);
      \draw[dashed,red,thick] (0,6.5) .. controls (-3,6.5) and (-2,2) .. (q);

    \end{tikzpicture}
  \end{center}
  \vspace*{-18pt}
  \caption{Of the two possibilities of 
   drawing new edge $(p,q)$, only solid green is valid.}
  \label{deg2}
\end{figure} 

It follows that when the new edge connects two vertices of degree 1, additional
information (such as relative order of the vertices on the outer face) is necessary.


We present our streamed drawing algorithm 
for an outerplanar graph, $G$, in terms of placing vertices of $G$ on a circle.
We will later derive an algorithm for drawing $G$
using grid points. 
We show in Algorithm~\ref{alg_outerplanar}
how to handle adding a new edge to the graph.

\begin{algorithm}[hbt]
  \begin{algorithmic}[1]
    \ENSURE{vertices are placed on the circle in the same order as
            they appear on the outer face of the drawing}
    \REQUIRE{outerplanar drawing of graph $G$ on a circle; $e = (p,q)$ -- edge to be added}
    \IF{$q$ is a new vertex}
      \STATE place $q$ on the circle according to ordering that includes $e$
    \ELSE
      \STATE add a \emph{virtual} arc $e'$ connecting $p$ and $q$ according
        to the specification of $e$ s.t. $e'$ does not intersect any existing edge
      \STATE calculate order $O$ of vertices on the outer plane (taking $e'$
        into account)
      \STATE move vertices into place on the circle according to $O$
    \ENDIF
    \STATE draw $e$
  \end{algorithmic}
  \caption{Adding new edge to outerplanar drawing of graph $G$}
  \label{alg_outerplanar}
\end{algorithm}

As mentioned previously, maintaining the invariant guarantees planarity
of the drawing. We now show that vertices can move into place without
causing any intersections in the process.
\begin{lemma}
Moving a vertex $v$ inside the circle along a trajectory that does not intersect
any edge non-incident to $v$ does not introduce any intersections and maintains
the
order of edges around $v$.
\label{lemma_no_intersection}
\end{lemma}
\begin{proof}
Consider the drawing with edges incident to $v$ removed (marked with dashed
lines in Fig.~\ref{moving_vertex}). The face to which $v$ belongs (limited
by edges and circle boundary) is the area where $v$ can move. Because
it is convex, as $v$ moves, its incident edges never intersect the
boundaries of the area (other than at their incident vertices), and the relative
order of the edges stays the same.
\qed
\end{proof}

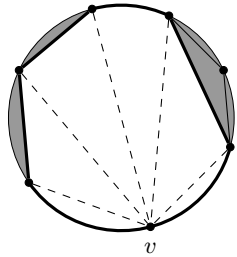
\begin{figure} 
  \begin{center}
    \begin{tikzpicture}
      [
        scale=0.75,
        n/.style={circle,fill=black,draw,inner sep=1pt},
        node distance=1pt
      ]

      \coordinate (ac) at ($(0,0)!1!15:(0,2)$);
      \coordinate (bc) at ($(0,0)!1!65:(0,2)$);
      \coordinate (cc) at ($(0,0)!1!125:(0,2)$);
      \coordinate (vc) at ($(0,0)!1!195:(0,2)$);
      \coordinate (dc) at ($(0,0)!1!255:(0,2)$);
      \coordinate (ec) at ($(0,0)!1!295:(0,2)$);
      \coordinate (fc) at ($(0,0)!1!335:(0,2)$);

      \fill[black!40,even odd rule] (0,0) circle (2)
                      (ac) -- (bc) -- (cc) arc (215:345:2) -- (fc) arc (65:105:2);

      \path[draw,name path=circ,ultra thin] (0,0) circle (2);

      \node[n] (a) at (ac) {};
      \node[n] (b) at (bc) {};
      \node[n] (c) at (cc) {};
      \node[n] (v) at (vc) {};
      \node[n] (d) at (dc) {};
      \node[n] (e) at (ec) {};
      \node[n] (f) at (fc) {};

      \node[below=of v] {$v$};

      \draw[very thick] (a) -- (b);
      \draw[very thick] (b) -- (c);
      \draw[very thick] (d) -- (f);
      \draw[very thick] (cc) arc (215:345:2);
      \draw[very thick] (fc) arc (65:105:2);
      \draw (d) -- (e);
      \draw (e) -- (f);
      \draw[dashed] (a) -- (v);
      \draw[dashed] (b) -- (v);
      \draw[dashed] (c) -- (v);
      \draw[dashed] (d) -- (v);
      \draw[dashed] (f) -- (v);
    \end{tikzpicture}
  \end{center}
  \vspace*{-12pt}
  \caption{Vertex $v$ can move in the (convex) white area without causing intersections
           or edge order changes.}
  \label{moving_vertex}
\end{figure} 

\begin{lemma}
A vertex, $v$, that moves into its new position can do so without crossing
any edges.
\end{lemma}
\begin{proof}
Consider the face (in the sense of Lemma~\ref{lemma_no_intersection}
and Fig.~\ref{moving_vertex}) of the drawing that $v$ belongs to.
The drawing is still outerplanar after adding the \emph{virtual} arc
(which is not necessarily straight-line), and therefore at least part of the
circle, $C'$, that forms this face's border still belongs to the outer face
of the drawing (see vertex $c$ in Fig.~\ref{adding_new_edge}).
By Lemma~\ref{lemma_no_intersection}, $v$ can move to $C'$
without crossing any edges.
When there are more vertices whose destination is the same part
of the circle, $C''$, (vertices $d$, $e$ and $f$ in Fig.~\ref{adding_new_edge}),
they must form a path with inner vertices all having degree 2.
After adding the new edge, their order on the circle
(and hence on $C''$) is the reverse of their current order, so their (straight-line)
trajectories do not cross. Since they form a path, edges between them will
not intersect as they move into place.
\qed
\end{proof}

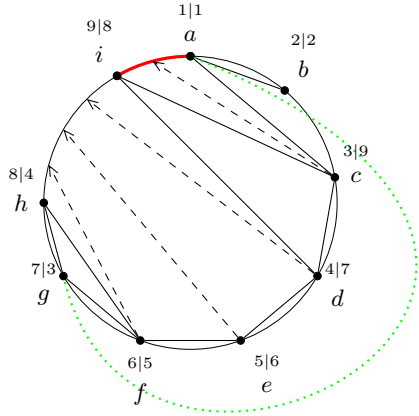
\begin{figure} 
  \vspace*{-8pt}
  \begin{center}
    \begin{tikzpicture}
      [
        scale=0.65,
        n/.style={circle,fill=black,draw,inner sep=1pt},
        node distance=1pt,
        every lower node part/.style={red}
      ]

      \useasboundingbox (-3,-4) rectangle (3,3);

      \draw[ultra thin] (0,0) circle (3);

      \node[n] (a) at ($(0,0)!1!0:(0,3)$) {};
      \node[n] (b) at ($(0,0)!1!320:(0,3)$) {};
      \node[n] (c) at ($(0,0)!1!280:(0,3)$) {};
      \node[n] (d) at ($(0,0)!1!240:(0,3)$) {};
      \node[n] (e) at ($(0,0)!1!200:(0,3)$) {};
      \node[n] (f) at ($(0,0)!1!160:(0,3)$) {};
      \node[n] (g) at ($(0,0)!1!120:(0,3)$) {};
      \node[n] (h) at ($(0,0)!1!90:(0,3)$) {};
      \node[n] (i) at ($(0,0)!1!30:(0,3)$) {};

      \draw[very thick,red] (a) arc (90:120:3);
      \node[n] at (a) {};
      \node[n] at (i) {};

      \coordinate (c_dest) at ($(0,0)!1!15:(0,3)$);
      \coordinate (d_dest) at ($(0,0)!1!45:(0,3)$);
      \coordinate (e_dest) at ($(0,0)!1!60:(0,3)$);
      \coordinate (f_dest) at ($(0,0)!1!75:(0,3)$);

      \draw (a) -- (b);
      \draw (a) -- (c);
      \draw (c) -- (d);
      \draw (d) -- (e);
      \draw (e) -- (f);
      \draw (f) -- (g);
      \draw (g) -- (h);
      \draw (h) -- (f);
      \draw (i) -- (c);
      \draw (i) -- (d);

      \draw[dotted,green,thick] (a) .. controls (11,-1.5) and (-1,-8) .. (g);

      \node[above=of a] (a1) {$a$};
      \node[above right=of b] (b1) {$b$};
      \node[right=of c] (c1) {$c$};
      \node[below right=of d] (d1) {$d$};
      \node[below right=of e] (e1) {\tiny $5|6$};
      \node[below=of f] (f1) {\tiny $6|5$};
      \node[below left=of g] (g1) {$g$};
      \node[left=of h] (h1) {$h$};
      \node[above left=of i] (i1) {$i$};

      \node[inner sep=0,above=of a1] {\tiny $1|1$};
      \node[inner sep=0,above=of b1] {\tiny $2|2$};
      \node[inner sep=0,above=of c1] {\tiny $3|9$};
      \node[inner sep=0,above=of d1] {\tiny $4|7$};
      \node[inner sep=0,below=of e1] {$e$};
      \node[inner sep=0,below=of f1] {$f$};
      \node[inner sep=0,above=of g1] {\tiny $7|3$};
      \node[inner sep=0,above=of h1] {\tiny $8|4$};
      \node[inner sep=0,above=of i1] {\tiny $9|8$};

      \draw[>=angle 45,->,dashed] (c) -- (c_dest);
      \draw[>=angle 45,->,dashed] (d) -- (d_dest);
      \draw[>=angle 45,->,dashed] (e) -- (e_dest);
      \draw[>=angle 45,->,dashed] (f) -- (f_dest);
    \end{tikzpicture}
  \end{center}
  \vspace*{-12pt}
  \caption{Adding edge $(a,g)$. \emph{Virtual} arc is drawn in green dots.
  Part of the circle that lies in the outer plane and is reachable from
  $c$ is shown in bold red.
  Numbers above vertices denote the order of vertex before and after
  the edge is added, respectively. Dashed lines are the trajectories
  of the vertices that need to move to maintain the invariant.}
  \label{adding_new_edge}
\end{figure} 

\begin{corollary}
Algorithm~\ref{alg_outerplanar} maintains an outerplanar drawing of a graph $G$
as new edges are added to it.
\label{ok_corollary}
\end{corollary}

Extending Algorithm~\ref{alg_outerplanar} to placing nodes on a grid
is straightforward. Instead of a circle, we operate on a set of grid
points in convex position that are approximately circular. 
We apply one of the algorithms for the file maintenance problem
or the online list labeling problem
for maintaining order of vertices in this set. 
When such an algorithm would
move vertex $v$, we first check if there is an unused grid point between
new neighbors of $v$ on the circle. If so, we simply move the vertex
to that point. Otherwise, we ``reserve'' the destination for $v$ by inserting
a stub vertex in the correct place (between new neighbors of $v$) on the circle.
The list labeling algorithm will move vertices around the circle 
(without changing order on the circle, so it will 
not cause any intersections) to make room for this stub.
Afterwards, we move $v$ into its reserved position.

\begin{lemma}
Vertex $v$ is moved in line~6 of Algorithm~\ref{alg_outerplanar} at most $deg(v)-1$ times.
\label{lemma_one_move}
\end{lemma}
\begin{proof}
A vertex $v$ is moved when the new edge forms a \emph{shortcut} that bypasses $v$
on the outer face. $v$ can appear at most $deg(v)$ times on the outer face, so
after $deg(v) - 1$ moves, there will be only one valid position for $v$ on the outer
face, so it cannot be bypassed anymore. (See Fig.~\ref{bypass}.)
\qed
\end{proof}

\begin{figure} 
  \begin{center}
    \begin{tikzpicture}
      [
        scale=0.75,
        n/.style={circle,fill=black,draw,inner sep=1pt},
        node distance=1pt,
      ]

      \node[n] (v) at (0,0) {};
      \node[above right=of v] {$v$};

      \coordinate (a) at ($(0,0)!1!120:(0,2)$);
      \coordinate (b) at ($(0,0)!1!240:(0,2)$);

      \draw (v) -- (0,2);
      \draw (v) -- (a);
      \draw (v) -- (b);

      \coordinate (v1) at ($(v)!0.7!(a)$);
      \coordinate (v2) at ($(v)!0.9!(b)$);

      \draw[thick,green,postaction={on each segment={mid arrow}}] (v2) -- (v1);

      \coordinate (v3) at ($(v)!0.7!(0,2)$);
      \coordinate (v4) at ($(v)!0.5!(b)$);

      \draw[thick,green,postaction={on each segment={mid arrow}}] (v3) -- (v4);

      \coordinate (v5) at ($(v)!0.4!(a)$);
      \coordinate (v6) at ($(v)!0.5!(0,2)$);

      \draw[dashed,thick,red,postaction={on each segment={mid arrow}}] (v5) -- (v6);
    \end{tikzpicture}
  \end{center}
  \vspace*{-12pt}
  \caption{Vertex $v$ has degree 3. After two shortcuts (solid green lines)
    around $v$ have been added, adding a third (dashed red line) would completely
    surround $v$, violating outerplanarity. Arrows show direction of edges
    on the outer plane.}
  \label{bypass}
\end{figure}
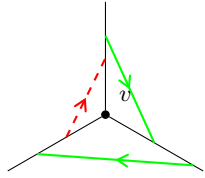 

\begin{theorem}
The grid-based version of
Algorithm~\ref{alg_outerplanar} 
maintains an outerplanar drawing of a graph $G$ and has the following
update performances:
uses $O(\log n)$
amortized moves per vertex, and 
\begin{enumerate}
\item
$O(n^3)$ area and $O(\log^2 n)$ vertex moves per edge insertion.
\item
polynomial area and $O(\log n)$ vertex moves per edge insertion.
\end{enumerate}
In addition, each of the above complexity bounds applies in an amortized
sense per vertex in the drawing.
\end{theorem}
\begin{proof}
By Corollary~\ref{ok_corollary}, 
we know that the algorithm maintains an outerplanar
drawing of $G$. 
For the area bounds, the file maintenance algorithms requires $O(n)$
available integer tags 
(in this case, points in convex position) to handle $n$ elements.
Since $m$ grid points in (strict) convex position require $O(m^3)$
area, the streamed drawing algorithm therefore uses $O(n^3)$ area in such 
cases.
Likewise, it uses polynomial area when using a solution to the online list
labeling problem.

With respect to the claim about amortized performance,
by Lemma~\ref{lemma_one_move}, each vertex $v$ is moved by
Algorithm~\ref{alg_outerplanar} at most $deg(v) - 1$ times.
Each such move requires at most one insertion into the list
for the file maintenance or list maintenance algorithm. 
This means that there are at most $O(n)$ such
insertions (sum of degrees of all vertices in an outerplanar graph is $O(n)$).
For $O(n)$ insertions, the performance for each of the file maintenance or list 
maintenance algorithms therefore become an
amortized number of moves per vertex made 
by our algorithm.
\qed
\end{proof}

Note that in this application we cannot immediately
apply our results for bulk moves,
unless we restrict our attention to possible vertex points 
that are uniformly distributed on a circle and moves that involve 
rotations of intervals of points around this circle.


\section{Conclusion}
In this paper, we provide a revised approach to streamed graph drawing
based on utilizing solutions to the file maintenance problem, either 
on a level-by-level basis (for level drawings of trees), a cross-product
basis (for tree-maps), or a circular/convex-position basis 
(for outerplanar graphs).
For future work,
it would be interesting to find other applications of this problem
in streamed or dynamic graph drawing applications.

\subsection*{Acknowledgements}
We would like to thank Alex Nicolau and Alex Veidenbaum
for helpful discussions regarding the file maintenance problem.
This work was supported in part by the NSF, under grants 1011840 and 1228639.

\bibliographystyle{abbrv}
\bibliography{refs}

\ifFull\else
\clearpage
\begin{appendix}
\section{Omitted Proof}
In this appendix, we include 
a proof that was omitted in the body of this paper.

\medskip
Specifically, we give the proof of Theorem~\ref{thm-lower}.

\bigskip
\noindent
\textbf{Theorem~\ref{thm-lower}:}~~Under 
the ordered streaming model without vertex moves, any tree drawing
algorithm requires $\Omega(2^{n/2})$ area in the worst case.

\end{appendix}
\fi

\end{document}